\documentclass[final,5p,times,9pt,twocolumn]{elsarticle}

\usepackage{lineno,hyperref}

\usepackage[cmex10]{amsmath}
\usepackage{amssymb}
\usepackage{mathrsfs}
\usepackage{amsthm}

\usepackage{graphicx}
\usepackage{color}
\usepackage{algorithm}
\usepackage{algorithmic}
\usepackage{bm}
\usepackage{setspace}

\DeclareGraphicsExtensions{.pdf,.png,.jpg,.eps,.ps}

\def\ba{\begin{array}}
	\def\ea{\end{array}}
\def\baa{\begin{align}}
	\def\eaa{\end{align}}

\newcommand{\bsq}{\begin{subequations}}
	\newcommand{\esq}{\end{subequations}}

\newcommand{\beq}{\begin{equation}}
\newcommand{\eeq}{\end{equation}}
\newcommand{\bq}{\begin{eqnarray}}
\newcommand{\eq}{\end{eqnarray}}
\newcommand{\bqn}{\begin{eqnarray*}}
	\newcommand{\eqn}{\end{eqnarray*}}
\newcommand{\bee}{\begin{enumerate}}
	\newcommand{\eee}{\end{enumerate}}
\newcommand{\bi}{\begin{itemize}}
	\newcommand{\ei}{\end{itemize}}

\usepackage{comment}
\usepackage{ifthen}
\newboolean{showcomments}
\setboolean{showcomments}{true}
\newcommand{\wang}[1]{\ifthenelse{\boolean{showcomments}}
	{ \textcolor[rgb]{1,0,1}{(ZW:  #1)}}{}}
\newcommand{\fliu}[1]{\ifthenelse{\boolean{showcomments}}
	{ \textcolor{red}{(FL:  #1)}}{}}
\newcommand{\slow}[1]{\ifthenelse{\boolean{showcomments}}
	{ \textcolor{blue}{(SL:  #1)}}{}}

\theoremstyle{definition}
\newtheorem{theorem}{Theorem}
\newtheorem{lemma}[theorem]{Lemma}

\theoremstyle{definition}

\newtheorem{remark}{Remark}

{\newtheorem{assumption}{\textit{Assumption}}}

\allowdisplaybreaks

\journal{xxx}

\begin{document}
\graphicspath{{Paper_Fig/}}
\setstretch{1}
\begin{frontmatter}

\title{Distributed Optimal Load Frequency Control Considering Nonsmooth Cost Functions}

\tnotetext[mytitlenote]{This work was supported  by the National Natural Science Foundation of China ( No. 51677100, U1766206, No. 51621065).}

\author[thu]{Zhaojian~Wang}
\author[thu]{Feng~Liu\corref{mycorrespondingauthor}}\ead{lfeng@mail.tsinghua.edu.cn}
\author[nrel]{Changhong Zhao}
\author[thu]{Zhiyuan Ma}
\author[thu]{Wei Wei}


\cortext[mycorrespondingauthor]{Corresponding author}

\address[thu]{State Key Laboratory of Power Systems, Department of Electrical Engineering, Tsinghua University, Beijing 100084, China}
\address[nrel]{National Renewable Energy Laboratory, Golden, CO, 80401, US}

\begin{abstract}
	This work addresses the distributed frequency control problem in power systems considering controllable load with a nonsmooth cost. The nonsmoothness exists widely in power systems, such as tiered price, greatly challenging the design of distributed optimal controllers. In this regard, we first formulate an optimization problem that minimizes the nonsmooth regulation cost, where both capacity limits of controllable load and tie-line flow are considered. Then, a distributed controller is derived using the Clark generalized gradient. We also prove the optimality of the equilibrium of the closed-loop system as well as its asymptotic stability. Simulations carried out on the IEEE 68-bus system verifies the effectiveness of the proposed method.
\end{abstract}

\end{frontmatter}

\section{Introduction}
With the proliferation of renewable generations, frequency control in power systems is facing a great challenge as power mismatch can fluctuate rapidly in a large amount. In this situation, the conventional centralized hierarchical control architecture may not respond fast enough due to large inertia of the traditional synchronous generators \cite{dorfler2016breaking,wang2019distributed}. On the other hand, load-side controllable resources with fast response capabilities provide a new opportunity to frequency regulation \cite{Schweppe1980Homeostatic}. In addition, as controllable loads are usually dispersed geographically vast across the power system, a distributed architecture is more desirable for load frequency control than the centralized one. 

Recently, the so-called \emph{reverse engineering} methodology is proposed by combining frequency control with optimal operation problems in power systems \cite{zhang:real, Stegink:aunifying, Li:Connecting, Cai:Distributed, Changhong:Design}. Under this framework, distributed load frequency control is widely investigated \cite{Changhong:Design, mallada2017optimal, Distributed_I:Wang, Kasis:Primary1, Distributed_II:Wang,zhao2018distributed, wang2018distributed}. In \cite{Changhong:Design}, an optimal load frequency control problem is formulated and a distributed controller is derived using controllable loads to realize primary frequency control. To eliminated the frequency deviation, the method is further extended in \cite{mallada2017optimal, zhao2018distributed} to realize a secondary load frequency control. At the same time, the tie-line power limit is considered. The design approach is generalized in \cite{Kasis:Primary1}, where the specific model requirement is eliminated. It only requires that the bus dynamics satisfy a passivity condition to guarantee asymptotic stability. In \cite{Distributed_I:Wang,Distributed_II:Wang}, the operational constraints including regulation capacity limits and tie-line power limits are considered, which guarantee both steady-state and transient capacity limit constraints. In \cite{wang2018distributed}, the distributed load frequency control under time-varying and unknown power injection is investigated, which can recover the nominal frequency even under unknown disturbances. 
The distributed load frequency control is of course a paid service, i.e., the system operator needs to pay for the controllable load to regulate their power. In the existing literature, the cost of controllable load is assumed to be differentiable, or equivalently, the price of the controllable load is continuous. This is not true for a variety of cases, e.g., the price may have step changes when controllable load values are in different intervals. In such a situation, the regulation is inherently nonsmooth, which makes existing methods difficult to apply.

This work designs a distributed controller for the optimal load frequency control in power systems, where the regulation cost function can be nonsmooth. We relax the assumption of the objective function from being differentiable to nonsmooth. This work is partly motivated by \cite{zeng2018distributed}. However, different from it, we consider the interplay between the solving algorithm and the power system dynamics and prove the stability of the closed-loop system. Another difference is that the objective function in \cite{zeng2018distributed} is strictly convex with respect to all decision variables. That is not necessary in our work, where some variables may not appear in the objective function. In such a situation, we prove the asymptotic convergence of the closed-loop system as well as the optimality of equilibrium.  

The rest of this paper is organized as follows. In Section II, we introduce some preliminaries and system models. Section III formulates the optimal load frequency control problem and introduces the distributed controller. In Section IV, convergence of the closed-loop system and optimality of the equilibrium point are proved. We confirm the performance of the controller via simulations on IEEE 68-bus system in Section V. Section VI concludes the paper.

\section{Problem Description}
\subsection{Preliminaries and notations}
\subsubsection{Notations}
In this paper, use $\mathbb{R}^n$ ($\mathbb{R}^n_{+}$) to denote the $n$-dimensional (nonnegative) Euclidean space. For a column vector $x\in \mathbb{R}^n$ (matrix $A\in \mathbb{R}^{m\times n}$), $x^{\mathrm{T}}$($A^{\mathrm{T}}$) denotes its transpose. For vectors $x,y\in \mathbb{R}^n$, $x^{\mathrm{T}}y=\left\langle x,y \right\rangle$ denotes the inner product of $x,y$. 
$\left\|x \right\|=\sqrt{x^{\mathrm{T}}x}$ denotes the Euclidean norm of $x$. 
Use $\textbf{1}$ to denote the vector with all $1$ elements. For a matrix $A=[a_{ij}]$, $a_{ij}$ stands for the entry in the $i$-th row and $j$-th column of $A$. Use $\prod_{i=1}^n\Omega_i$ to denote the Cartesian product of the sets $\Omega_i, i=1, \cdots, n$. 
Given a collection of $y_i$ for $i$ in a certain set $Y$, $y$ denotes the column vector
$y := (y_i, i\in Y)$ with a proper dimension, and $y_i$ as its components.

\subsubsection{Preliminaries}
Let $ f(x):\mathbb{R}^n\rightarrow\mathbb{R} $ be a locally Lipschitz continuous function and denote its Clarke generalized gradient by $\partial f(x)$ \cite[Page 27]{clarke:optimization}. For a continuous strictly convex function $ f(x):\mathbb{R}^n\rightarrow\mathbb{R} $, we have $ (g_x-g_y)^{\mathrm{T}}(x-y)>0,\ \forall x\neq y $, where $g_x\in \partial f(x)$ and $g_y\in \partial f(y)$.

Define the projection of $x$ onto a closed convex set $\Omega$ as 
\begin{eqnarray}
\label{def_projection}
\mathcal{P}_{\Omega}(x)=\arg \min\nolimits_{y\in \Omega}\left\|x-y \right\|
\end{eqnarray}
Use ${\rm{Id}}$ to denote the identity operator, i.e., ${\rm{Id}}(x)=x$, $\forall x$. Define $N_\Omega(x)=\{v|\left\langle v, y-x\right\rangle\le 0, \forall y\in \Omega\}$. We have $\mathcal{P}_{\Omega}(x)=({\rm{Id}}+N_{\Omega})^{-1}(x)$ \cite[Chapter 23.1]{bauschke2011convex}.

A basic property of a projection is 
\begin{eqnarray}
\label{Basic_projection}
(x-\mathcal{P}_{\Omega}(x))^{\rm T}(y-\mathcal{P}_{\Omega}(x))\le 0,\ \forall x\in\mathbb{R}^n,\  y\in \Omega 
\end{eqnarray} 
Moreover, we also have \cite[Theorem 1.5.5]{facchinei2003finite}
\begin{eqnarray}
\label{Property1_projection}
(\mathcal{P}_{\Omega}(x)-\mathcal{P}_{\Omega}(y))^{\rm T}(x-y)\ge \left\|\mathcal{P}_{\Omega}(x)-\mathcal{P}_{\Omega}(y) \right\|^2
\end{eqnarray}

Define  $ V(x):=\frac{1}{2}\big(\left\| x - \mathcal{P}_{\Omega}(y)\right\|^2-\left\| x - \mathcal{P}_{\Omega}(x)\right\|^2\big) $, and then $ V(x) $ is differentiable and convex with respect to $ x $ \cite[Lemma 4]{liu2013one}. Moreover, we have
\begin{align}
\label{Property2_projection}
V(x)&=\frac{1}{2}\left\|\mathcal{P}_{\Omega}(x)-\mathcal{P}_{\Omega}(y) \right\|^2 \nonumber\\
&\qquad\quad -(x-\mathcal{P}_{\Omega}(x))^{\rm T}(\mathcal{P}_{\Omega}(y)-\mathcal{P}_{\Omega}(x))\\
\label{Property2_projection2}
&\ge \frac{1}{2}\left\|\mathcal{P}_{\Omega}(x)-\mathcal{P}_{\Omega}(y) \right\|^2\ge 0 \\
\label{Property3_projection}
\nabla V(x)&=\mathcal{P}_{\Omega}(x)-\mathcal{P}_{\Omega}(y)
\end{align}
where the inequality is due to \eqref{Basic_projection}. From \eqref{Property2_projection}, $ V(x)=0 $ holds only when $ \mathcal{P}_{\Omega}(x)=\mathcal{P}_{\Omega}(y) $.

\subsection{Network model}
A power network is usually composed of multiple buses, which are connected with each other through transmission lines. It can be modeled as a graph $\mathcal{G}:=(\mathcal{N}, \mathcal{E})$, where  $\mathcal{N}=\{0,1,2,...n\}$ is the set of buses and
$\mathcal{E}\subseteq \mathcal{N}\times \mathcal{N}$ is the set of edges (transmission lines). 
Let $m= |\mathcal{E}|$ denote the number of lines.
The buses are divided into two types: generator buses, denoted by $\mathcal{N}_g$ and load buses, denoted by $\mathcal{N}_l$. A generator bus contains a generator ( possibly with certain aggregate load). A load bus has only load with no generator.  
The graph $\mathcal{G}$ is treated as directed with an arbitrary orientation and use $(i,j)\in \mathcal{E}$
or $i\rightarrow j$ interchangeably to denote a directed edge from $i$ to $j$. 
Without loss of generality, we assume the graph is connected
and node $0$ is a reference node. The incidence matrix of the graph is denoted by $C$, and we have $\textbf{1}^{\rm T}C=0$.

We adopt a second-order linearized model to describe the frequency dynamics of each bus. We assume that the  lines are lossless and adopt the DC power flow model \cite{Distributed_I:Wang, Changhong:Design}. For each bus $j\in \mathcal{N}$, let $\theta_j(t)$ denote the rotor angle at node $j$ at time $t$ and $\omega_j(t)$ the frequency. \footnote{Sometimes, we also omit $t$ for simplicity.} Let $P^l_j(t)$ denote the controllable load. Let given constant $P^m_j$ denote any change in power injection, that occurs on the generation side or the load side, or both. Define $ \theta_{ij}=\theta_{i}-\theta_{j} $ as the angle difference between bus $ i $ and $ j $, and its compact form is denoted by $ \theta_{e}=(\theta_{ij}, (i,j)\in\mathcal{E}) $.
Then for each node $j\in \mathcal{N}$, the dynamics are
\begin{subequations}
	\begin{align}
	\dot \theta_{ij} & = \omega_i-\omega_j, \quad j\in \mathcal{N}
	\label{eq:model.1a}
	\\
	\dot \omega_j & = \frac{1}{M_j} \left( P^m_j - P^{l}_j  -D_j \omega_j \right. \nonumber\\
	& \quad  + \sum\nolimits_{i: i\rightarrow j} B_{ij}\theta_{ij} 
	-  \sum\nolimits_{k: j\rightarrow k} B_{jk}\theta_{jk}\bigg), \  j\in \mathcal{N}_g
	\label{eq:model.1b}
	\\
	0 & =   P^m_j - P^{l}_j -D_j \omega_j  \nonumber\\
	& \quad + \sum\nolimits_{i: i\rightarrow j} \! B_{ij}\theta_{ij} 
	-  \sum\nolimits_{k: j\rightarrow k} \! B_{jk}\theta_{jk},   j\in \mathcal{N}_l
	\label{eq:model.1c}
	\end{align}
	\label{eq:model.1}where $M_j>0$ are inertia constants, $D_j>0$ are damping constants, and $B_{jk}>0$ are line parameters that depend on the reactance of the line $(j,k)$. 
\end{subequations}

The scenario is that: the system operates in a steady state at first. A certain power imbalance occurs due to variation of power injection $P_j^m$. Then controllable load accordingly changes its output to eliminate the imbalance. 
%
%
%
%

\section{Problem Formulation}
In this section, we first formulate the optimal load frequency problem with a nonsmooth objective function. Then, we propose a distributed controller based on the Clark generalized gradient to drive the power system to the optimal solution. 
\subsection{Optimization problem}
The optimization problem is 
\begin{subequations}
	\setlength{\abovedisplayskip}{4pt}	
	\setlength{\belowdisplayskip}{4pt}
	\label{OLC}     
	\begin{align}
	\min\limits_{P_j^l,\ \phi_{j}}\quad& f(P^l)= \sum\nolimits_{j\in \mathcal{N}} f_j(P_j^l)
	\label{OLC_1}
	\\ 
	\text{s.t.} \quad
	& 0=P_j^l - P_j^m - \sum\nolimits_{i:i\rightarrow j} B_{ij}(\phi_{i}-\phi_{j}) \nonumber\\
	&\qquad\quad+\sum\nolimits_{k:j\rightarrow k} B_{jk}(\phi_{j}-\phi_{k})  , \ j\in \mathcal{N}
	\label{OLC_2}
	\\
	\label{OLC_3}
	&\underline P_j^l\le P_j^l\le \overline P_j^l, \quad j\in \mathcal{N} \\
	\label{OLC_4}
	&{\underline\theta}_{ij} \le  \phi_{i}-\phi_{j} \le {\overline\theta}_{ij},\quad (i,j) \in \mathcal{E} 
	\end{align}
\end{subequations}
where $ \underline P_j^l\le \overline P_j^l $ are constants, denoting the lower and upper bound of $ P_j^l $. $ {\underline\theta}_{ij} \le {\overline\theta}_{ij} $ are also constants, denoting the lower and upper bound of angle difference. 
The first constraint is the local power balance. $ \phi_{j} $ is the virtual phase angle, which equals to $ \theta_{j} $ at the optimal solution. 
Use $ \phi_{ij}=\phi_{i}-\phi_{j} $ to denote the virtual phase angle difference. In the DC power flow, we have $ P_{ij}=B_{ij}\theta_{ij} $, where $ P_{ij} $ is the power of line $(i,j)$. Thus, \eqref{OLC_4} is in fact the tie-line power limit constraint.
We have the following assumptions.
\begin{assumption}\label{convex}
	$ f_j(P_j^l) $ is strictly convex.
\end{assumption}

\begin{assumption}
	\label{Slater}
	The Slater's condition  \cite[Chapter 5.2.3]{boyd2004convex} of \eqref{OLC} holds, i.e., problem \eqref{OLC} is feasible provided that the constraints are affine. 
\end{assumption}

\begin{remark}
	Assumption \ref{convex} could be further relaxed, since a non-strictly convex function can be strictly convexified by using a  nonlinear perturbation \cite{mangasarian1979nonlinear}.
\end{remark}

\begin{remark}\label{nonsmooth}
	Problem \eqref{OLC} allows the cost function  $ f_j(P_j^l) $ to be nonsmooth, which is required to be differentiable in the existing literature \cite{Changhong:Design, mallada2017optimal, Distributed_I:Wang, Kasis:Primary1, Distributed_II:Wang,zhao2018distributed, wang2018distributed}. Thus, the problem \eqref{OLC} is more general and suitable for a variety of real problems whose regulation costs are not smooth. 
	A typical example is the tiered price, where the price discontinuously increases with respect to the amount of controllable load. 
	It also should be noted that the decision variable $ \phi_{j} $ is absent in the objective function of \eqref{OLC}. It makes the paper not a trivial application of \cite{ zeng2018distributed}, i.e., the objective function is not required to be strictly convex to all the decision variables. It makes the convergence proof more challenging.
\end{remark}

\begin{remark}\label{differential_inclusion}
	In the existing literature, the controller usually involves the projection of a gradient onto a convex set. If the objective function is nonsmooth, it becomes the projection of a subdifferential set onto a convex set. In this situation, the existence of trajectories is not guaranteed \cite{zeng2018distributed,zhou2019adaptive}, which makes existing load frequency control methods inapplicable to the nonsmooth case.
\end{remark}

\subsection{Controller Design}
To help the controller design, we make a modification on the problem \eqref{OLC}.
\begin{subequations}
	\setlength{\abovedisplayskip}{4pt}	
	\setlength{\belowdisplayskip}{4pt}
	\label{OLC2.0}     
	\begin{align}
	\min\limits_{P_j^l,\ \phi_{j}}\quad& f(P^l)= \sum\nolimits_{j\in \mathcal{N}} f_j(P_j^l) + \frac{1}{2}\sum\nolimits_{j\in \mathcal{N}} z_j^2
	\label{OLC2.1}
	\\ 
	\text{s.t.} \quad &
	\eqref{OLC_2}, \eqref{OLC_3}, \eqref{OLC_4}
	\end{align}
\end{subequations}
where $ z_j= P_j^l - P_j^m - \sum\nolimits_{i:i\rightarrow j} B_{ij}\phi_{ij} +\sum\nolimits_{k:j\rightarrow k} B_{jk}\phi_{jk} $. For any feasible solution to \eqref{OLC}, $ z_j=0 $. Thus, \eqref{OLC} and \eqref{OLC2.0} have same solutions. 

Define the sets 
\begin{eqnarray}
\setlength{\abovedisplayskip}{4pt}	
\setlength{\belowdisplayskip}{4pt}
\Omega_j:=\left\{P_j^l\ |\ \underline P_j^l\le P_j^l\le \overline P_j^l\right\},\ \Omega=\prod\nolimits_{j=1}^n\Omega_j
\end{eqnarray}
Then, we give the controller for each controllable load, which is denoted by OLC. 
\begin{subequations}
	\setlength{\abovedisplayskip}{4pt}	
	\setlength{\belowdisplayskip}{4pt}
	\label{Controller}
	\begin{align}
	\label{Controller1}
	\dot {{d}}_{j}&\in\left\{p:p=-{{d}}_{j} + P_j^l + \omega_j -g_j(P_j^l) -z_j -\mu_j, \right.\nonumber\\
	& \quad\left. g_j(P_j^l)\in \partial f_j(P_j^l)\right\}  \\
	\label{Controller2}
	\dot \mu_{j}&= P_j^l - P_j^m - \sum\limits_{i:i\rightarrow j} B_{ij}\phi_{ij} +\sum\limits_{k:j\rightarrow k} B_{jk}\phi_{jk}   \\
	\label{Controller3}
	\dot \phi_{j}&= \sum\limits_{i:i\rightarrow j} B_{ij}(\mu_i-\mu_j) - \sum\limits_{k:j\rightarrow k} B_{jk}(\mu_j-\mu_k) \nonumber\\
	&\ -\sum_{(i,j)\in \cal E} \eta^-_{ij}+\sum_{(j,k)\in \cal E} \eta^-_{jk} +\sum_{(i,j)\in \cal E} \eta^+_{ij}-\sum_{(j,k)\in \cal E} \eta^+_{jk}\nonumber\\
	& \ + \sum\limits_{i:i\rightarrow j} B_{ij}(z_i-z_j) - \sum\limits_{k:j\rightarrow k} B_{jk}(z_j-z_k)\\
	\label{Controller4}
	\dot {\varphi}^+_{ij}&=-{\varphi}^+_{ij}+{\eta}^+_{ij}+\phi_{ij}-\overline{\theta}_{ij}
	\\
	\label{Controller5}
	\dot {\varphi}^-_{ij}&=-{\varphi}^-_{ij}+{\eta}^-_{ij}+\underline{\theta}_{ij}-\phi_{ij}
	\\
	\label{Controller6}
	P_j^l &= \mathcal{P}_{\Omega_j}\left(d_{j}\right)\\ 
	\label{Controller7}
	{\eta}^+_{ij} &= \mathcal{P}_{\mathbb{R}_+}\left({\varphi}^+_{ij}\right)\\
	\label{Controller8} 
	{\eta}^-_{ij} &= \mathcal{P}_{\mathbb{R}_+}\left({\varphi}^-_{ij}\right)
	\end{align}
\end{subequations}
Combining with the power system dynamics, we have the closed-loop system \eqref{eq:model.1}, \eqref{Controller}.

\begin{remark}[Load demand estimate]\label{load_estimate}
	In power systems, the load demand $ P_j^l $ is  difficult to measure. Similar to  \cite{mallada2017optimal,zhao2018distributed,Distributed_II:Wang}, $ P_j^l - P_j^m $ in \eqref{Controller2} can be substituted equivalently in following ways. For $ j\in \mathcal{N}_g $, 
	\begin{equation*}
	\setlength{\abovedisplayskip}{4pt}	
	\setlength{\belowdisplayskip}{4pt}
	P_j^l - P_j^m = -M_j\dot \omega_j  -D_j \omega_j + \sum\nolimits_{i: i\rightarrow j} P_{ij}
	-  \sum\nolimits_{k: j\rightarrow k} P_{jk}
	\end{equation*}
	For $ j\in \mathcal{N}_l $, 
	\begin{equation*}
	\setlength{\abovedisplayskip}{4pt}	
	\setlength{\belowdisplayskip}{4pt}
	P_j^l - P_j^m = -D_j \omega_j + \sum\nolimits_{i: i\rightarrow j} P_{ij} 
	-  \sum\nolimits_{k: j\rightarrow k} P_{jk}
	\end{equation*}
	In this way, the measurement of load demand $ P_j^l $ is avoided. We only need to measure $\omega_j, P_{ij}$, which are much easier to realize. Moreover, the power loss  can be treated as unknown load demand, which can be also considered by this method. 
\end{remark}

\section{Optimality and Convergence}
In this section, we address the optimality of the equilibrium point and the convergence of the closed-loop system.
\subsection{Optimality}
Denote $x=(\theta, \omega_g, d, \mu, \phi, \varphi^-, \varphi^+)$ and $y=(x, P^{l}, {\eta}^{+}, {\eta}^{-})$. Let $x^*=(\theta^*, \omega_g^*, d^*, \mu^*,$  $ \phi^*,$ $ \varphi^{-*}, \varphi^{+*})$ be an equilibrium of the closed-loop system \eqref{eq:model.1}, \eqref{Controller}. Then, there exists $ g(P^{l*})\in \partial f(P^{l*})$ such that
\begin{subequations}
	\setlength{\abovedisplayskip}{4pt}	
	\setlength{\belowdisplayskip}{4pt}
	\label{equilibrium}
	\begin{align}
	0& =C^{\rm T} \omega^*
	\label{equilibrium1}
	\\
	0& =  P^{m} - P^{l*}  -D \omega^* - CBC^{\rm T}\theta^*
	\label{equilibrium2}\\
	\label{equilibrium3}
	0&=-{{d}}^{*} + P^{l*} + \omega^* -g(P^{l*}) -\mu^*  \\
	\label{equilibrium4}
	0&= P^{l*} - P^m +  CBC^{\rm T}\phi^*   \\
	\label{equilibrium5}
	0&=  -CBC^{\rm T}\mu^* -C  \eta^{-*}+C\eta^{+*} \\
	\label{equilibrium6}
	0&=-{\varphi}^{+*} +{\eta}^{+*} -C^{\rm T}\phi^*-\overline{\theta}
	\\
	\label{equilibrium7}
	0&=-{\varphi}^{-*}+{\eta}^{-*}+\underline{\theta}+C^{\rm T}\phi^*\\
	\label{equilibrium8}
	P^{l*} &= \mathcal{P}_{\Omega}\left(d^*\right)\\
	\label{equilibrium9} 
	{\eta}^{+*} &= \mathcal{P}_{\mathbb{R}^{m}_+}\left({\varphi}^{+*}\right)\\
	\label{equilibrium10}
	{\eta}^{-*} &= \mathcal{P}_{\mathbb{R}^{m}_+}\left({\varphi}^{-*}\right)
	\end{align}
\end{subequations}

Now, we introduce the properties of the equilibrium points. 
\begin{theorem}\label{Equivalence}
	Suppose Assumptions \ref{convex} and \ref{Slater} hold. We have 
	\begin{enumerate}
		\item The nominal frequency is restored, i.e., $\omega^*_j=0$ for all $j\in \mathcal{N}$.
		\item If $ x^* $ is an equilibrium point of \eqref{eq:model.1}, \eqref{Controller}, then $ (P^{l*}, \phi^*) $ is an optimal solution to  \eqref{OLC} and $ (\mu^*, {\eta}^{+*}, {\eta}^{-*}) $ is an optimal solution to its dual problem.
		\item $\phi^*_{ij} = \theta^*_{ij}$ for all $(i,j)\in \mathcal{E}$. Moreover, the line limits are satisfied by $x^*$, implying $\underline{\theta}_{ij} \le \theta^*_{ij} \le \overline{\theta}_{ij}$ on every tie line $(i,j)\in \mathcal{E}$.
		\item At the equilibrium, 
		$(\theta^*, \phi^*, \omega_g^*, P^{l*})$ is unique, with $(\theta^*, \phi^*)$ being unique up to (equilibrium) reference angles $(\theta_0, \phi_0)$.
	\end{enumerate}
\end{theorem}
\begin{proof}
	1) From \eqref{equilibrium2} and \eqref{equilibrium4}, we have $ \textbf{1}^{\rm T}D\omega^*=0 $. From \eqref{equilibrium1}, we have $ \omega^*= \omega_0\cdot\textbf{1}$ with a constant $ \omega_0 $. As $D$ is a diagonal positive definite matrix, we have $ \omega_0=0 $.
	
	2) From \eqref{equilibrium3} and \eqref{equilibrium6}-\eqref{equilibrium10}, we have 
	\begin{subequations}\label{KKT0}
		\setlength{\abovedisplayskip}{4pt}	
		\setlength{\belowdisplayskip}{4pt}
		\begin{align}
		\label{KKT1}
		P^{l*} &= \mathcal{P}_{\Omega}\left( P^{l*} -g(P^{l*}) -\mu^*\right)\\
		\label{KKT2}
		{\eta}^{+*} &= \mathcal{P}_{\mathbb{R}^{m}_+}\left({\eta}^{+*} -C^{\rm T}\phi^*-\overline{\theta}\right)\\
		\label{KKT3}
		{\eta}^{-*} &= \mathcal{P}_{\mathbb{R}^{m}_+}\left({\eta}^{-*}+\underline{\theta}+C^{\rm T}\phi^*\right)
		\end{align}
	\end{subequations}
	or equivalently, 
	\begin{subequations}\label{KKT2.0}
		\setlength{\abovedisplayskip}{4pt}	
		\setlength{\belowdisplayskip}{4pt}
		\begin{align}
		\label{KKT2.1}
		-g(P^{l*}) -\mu^*&\in N_{\Omega}(P^{l*})\\
		\label{KKT2.2}
		-C^{\rm T}\phi^*-\overline{\theta}&\in N_{\mathbb{R}^{m}_+}({\eta}^{+*})\\
		\label{KKT2.3}
		\underline{\theta}+C^{\rm T}\phi^*&\in N_{\mathbb{R}^{m}_+}({\eta}^{-*})
		\end{align}
	\end{subequations}
	By the KKT condition in \cite[Theorem 3.34]{ruszczynski2006nonlinear}, \eqref{equilibrium4}, \eqref{equilibrium5} and \eqref{KKT2.0} coincide with the KKT optimality condition of the problem \eqref{OLC}. Then, we have this assertion.
	
	3) From \eqref{equilibrium2} and \eqref{equilibrium4}, we have $ CBC^{\rm T}(\theta^*-\phi^*)=0 $, which holds for any incidence matrix $ C $. Thus, we have $ \theta^*-\phi^*=c_0\cdot\textbf{1} $ with a constant $c_0$. Then, we have $ \theta_{ij}^*-\phi_{ij}^*=0 $. Moreover, by 2), we know $\underline{\theta}_{ij} \le \phi_{ij}^*\le \overline{\theta}_{ij}$, which implies that $\underline{\theta}_{ij} \le \theta_{ij}^*\le \overline{\theta}_{ij}$.	
	
	4) $ P^{l*}$ is unique because the objective function in \eqref{OLC_1} is strictly convex in $P^{l}$. $\omega^*$ is unique due to $\omega^*=0$. By \eqref{equilibrium4}, we know $\phi^*$ is unique modulo a rigid (uniform) rotation of all angles. Since $\theta^* - \phi^*=c_0\cdot\textbf{1} $, it implies that $\theta^*$ is also unique modulo a rigid rotation. 
	This proves the uniqueness of $(\theta^*, \phi^*, \omega_g^*, P^{l*})$.
\end{proof}

\subsection{Convergence}
Define the function
\begin{equation}\label{V2}
\setlength{\abovedisplayskip}{4pt}	
\setlength{\belowdisplayskip}{4pt}
V(x)=V_1(x)+V_2(x)
\end{equation} 
where 
\begin{align}\label{V_1}
&V_1(x)= \frac{1}{2}\left\| P^l - P^{l*}\right\|^2 + \frac{1}{2}\left\|\mu - \mu^*\right\|^2+ \frac{1}{2}\left\|\phi - \phi^*\right\|^2  \nonumber \\
& \quad  + \frac{1}{2}\left\| \eta^{+} - \eta^{+*}\right\|^2 + \frac{1}{2}\left(\theta_e - \theta_e^*\right)^{\rm T}B\left(\theta_e - \theta_e^*\right)  \nonumber\\
&\quad + \frac{1}{2}\left\| \eta^{-} - \eta^{-*}\right\|^2  + \frac{1}{2}\left(\omega_g - \omega_g^*\right)^{\rm T}M\left(\omega_g - \omega_g^*\right) 
\end{align}
\begin{align}
\label{V_2}
&V_2(x)=-(d-P^l)^{\rm T}(P^{l*}-P^l)-(\varphi^+-\eta^+)^{\rm T}(\eta^{+*}-\eta^+) \nonumber\\
&\qquad\quad\  -(\varphi^--\eta^-)^{\rm T}(\eta^{-*}-\eta^-)
\end{align}      

Then, we have the following result about $ V(x) $.
\begin{lemma}\label{V_dot}
	Suppose Assumptions \ref{convex} and \ref{Slater} hold. Then the function $ V(x) $ has following properties
	\begin{enumerate}
		\item $V(x)\ge0$ and $V(x)=0$ holds only at the equilibrium point.
		\item The time derivative of $V(x(t))$ satisfies $ \dot V(x(t))\leq 0  $.	
	\end{enumerate}
\end{lemma}
\begin{proof}
	1) By \eqref{Basic_projection}, we know that $ V_2(x)\ge 0 $.                             
	From \eqref{V2}, \eqref{V_1} and \eqref{V_2}, we know $V(x)\ge0$ and  $V(x)=0$ holds only at the equilibrium point.
	
	2) By \eqref{Property3_projection}, the gradient of $V$ is
	\begin{equation}
	\setlength{\abovedisplayskip}{4pt}	
	\setlength{\belowdisplayskip}{4pt}
	\nabla V
	=\begin{bmatrix}
	\nabla_{ d} V \\
	\nabla_{\mu} V \\
	\nabla_{\phi} V \\
	\nabla_{\eta^+} V \\
	\nabla_{\eta^-} V \\
	\nabla_{\theta} V \\
	\nabla_{\omega_g} V
	\end{bmatrix}
	=\begin{bmatrix}
	P^l-P^{l*} \\
	\mu - \mu^* \\
	\phi - \phi^* \\
	\eta^+ - \eta^{+*}\\
	\eta^- - \eta^{-*}\\
	B\left(\theta_e - \theta_e^*\right) \\
	M\left(\omega_g - \omega_g^*\right)
	\end{bmatrix}
	\end{equation}  
	
	Then, there is $ g(P^{l})\in \partial f(P^{l})$ such that the time derivative of $V$ is 
	\begin{align} \label{dotV1}
	&\dot V = (P^l-P^{l*})^{\rm T}(-{{d}} + P^l +\omega -g(P^l)-z -\mu) \nonumber\\
	& +(\mu - \mu^*)^{\rm T}(P^l-P^m+CBC^{\rm T}\phi)+ \left(\theta_e - \theta_e^*\right)^{\rm T}BC^{\rm T}\omega \nonumber\\
	&+ \left(\phi - \phi^*\right)^{\rm T} \left(-CBC^{\rm T}\mu-C\eta^-+C\eta^+-CBC^{\rm T}z\right) \nonumber\\
	&  + \left(\eta^+ - \eta^{+*}\right)^{\rm T}\left(-{\varphi}^+ +{\eta}^+ -C^{\rm T}\phi -\overline{\theta} \right)  \nonumber\\
	& + \left(\eta^- - \eta^{-*}\right)^{\rm T}\left(-{\varphi}^- +{\eta}^- +\underline{\theta} +C^{\rm T}\phi \right) \nonumber\\
	&  + \left(\omega - \omega^*\right)^{\rm T}(P^{m} - P^l  -D \omega-CB\theta_e)
	\end{align}
	where the last item is due to  the fact that, for each $ j\in \mathcal{N}_l $ 
	\begin{align}\label{load_freq}
	\setlength{\abovedisplayskip}{4pt}	
	\setlength{\belowdisplayskip}{4pt}
	0&=\left(\omega_j-\omega_j^*\right)\left(P^m_j - P^{l}_j -D_j \omega_j \right. \nonumber\\
	& \qquad \qquad + \sum\limits_{i: i\rightarrow j}  B_{ij}\theta_{ij} 
	-  \sum\limits_{k: j\rightarrow k}  B_{jk}\theta_{jk}\bigg)
	\end{align}

	Combing \eqref{dotV1} and \eqref{equilibrium}, we have 
	\begin{subequations}\label{dotV2}
		\setlength{\abovedisplayskip}{4pt}	
		\setlength{\belowdisplayskip}{4pt}
		\begin{align} 
		\dot V &= (\tilde P^l)^{\rm T}(-{\tilde {d}} +\tilde  P^l +\tilde \omega -g(P^l) +g(P^{l*}) -\tilde z-\tilde \mu)  \nonumber\\
		&+ \tilde\phi ^{\rm T} \left(-CBC^{\rm T}\tilde\mu-C\tilde\eta^-+C\tilde\eta^+-CBC^{\rm T}\tilde z\right) \nonumber\\
		&+\tilde \mu^{\rm T}(\tilde P^l +CBC^{\rm T}\tilde\phi) + \left(\tilde \eta^+ \right)^{\rm T}\left(-\tilde {\varphi}^+ +\tilde {\eta}^+ -C^{\rm T}\tilde \phi  \right) \nonumber\\
		& + \left(\tilde \eta^- \right)^{\rm T}\left(-\tilde {\varphi}^- +\tilde {\eta}^- +C^{\rm T}\tilde \phi \right) \nonumber\\
		&+\tilde \theta ^{\rm T}BC^{\rm T}\tilde \omega + \tilde \omega^{\rm T}( - \tilde P^l  -D\tilde  \omega-CB\tilde \theta)\nonumber\\
		&=- (P^l-P^{l*})^{\rm T}({{d}}-d^*) + \left\| P^l-P^{l*}\right\|^2 \label{dotV2.1} \\
		&-(\eta^+ -\eta^{+*})^{\rm T}({\varphi}^+-{\varphi}^{+*}) + \left\|\eta^+ -\eta^{+*}\right\|^2 \label{dotV2.2}\\
		&-(\eta^- -\eta^{-*})^{\rm T}({\varphi}^--{\varphi}^{-*}) + \left\|\eta^- -\eta^{-*}\right\|^2 \label{dotV2.3}\\
		&\label{dotV2.4} - (P^l-P^{l*})^{\rm T}\left(g(P^l) - g(P^{l*})\right)- \tilde \omega^{\rm T}D\tilde  \omega \\
		\label{dotV2.5}
		&-(\tilde P^l)^{\rm T}\tilde z - \tilde\phi ^{\rm T}CB C^{\rm T}\tilde z
		\end{align}
	\end{subequations}
	where $\tilde x=x-x^*$. 
	
	By the \cite[Theorem 1.5.5]{facchinei2003finite}, we have 
	\begin{equation}
	\setlength{\abovedisplayskip}{4pt}	
	\setlength{\belowdisplayskip}{4pt}
	- (P^l-P^{l*})^{\rm T}({{d}}-d^*) + \left\| P^l-P^{l*}\right\|^2\le 0
	\end{equation}
	Similarly, $ -(\eta^+ -\eta^{+*})^{\rm T}({\varphi}^+-{\varphi}^{+*}) + \left\|\eta^+ -\eta^{+*}\right\|^2 \le 0 $ and $ -(\eta^- -\eta^{-*})^{\rm T}({\varphi}^--{\varphi}^{-*}) + \left\|\eta^- -\eta^{-*}\right\|^2\le0 $ also hold.

	The convexity of $f$ implies that 
	\begin{equation}
	\setlength{\abovedisplayskip}{4pt}	
	\setlength{\belowdisplayskip}{4pt}
	- (P^l-P^{l*})^{\rm T}\left(g(P^l) - g(P^{l*})\right) \le 0
	\end{equation}
	We also have $-\tilde \omega^{\rm T}D\tilde  \omega\le0$ because $D$ is positive definite. In addition,
	\begin{equation}\label{dotV2.50}
	\setlength{\abovedisplayskip}{4pt}	
	\setlength{\belowdisplayskip}{4pt}
	-(\tilde P^l)^{\rm T}\tilde z - \tilde\phi ^{\rm T}CB C^{\rm T}\tilde z=-\tilde z^{\rm T}\cdot\tilde z \le 0
	\end{equation}
	Then, \eqref{dotV2.1}-\eqref{dotV2.5} are all nonpositive, i.e., $ \dot V(x(t))\leq 0  $.
\end{proof}

The following result shows the stability of the closed-loop system \eqref{eq:model.1}, \eqref{Controller}.
\begin{theorem}
	Suppose Assumptions \ref{convex} and \ref{Slater} hold. Then the trajectory of the closed loop system \eqref{eq:model.1}, \eqref{Controller} has following properties
	\begin{enumerate}
		\item $(x(t), P^{l}(t), \eta^+(t), \eta^-(t)) $ is bounded.
		\item $(x(t), P^{l}(t), \eta^+(t), \eta^-(t)) $  converges to  equilibrium of the closed-loop system \eqref{eq:model.1}, \eqref{Controller}.	
		\item The convergence of $x(t)$ is a point, i.e., $x(t)\rightarrow x^*$ as $t\rightarrow\infty$ for some equilibrium point $ x^* $.
	\end{enumerate}
\end{theorem}
\begin{proof}1) 
	From Lemma \ref{V_dot}, we know that $(\theta(t), \omega_g(t), \mu(t),$  $ \phi(t),$ $P^{l}(t), \eta^+(t), \eta^-(t) )$ is bounded. By \eqref{eq:model.1c}, $ \omega_l(t) $ is also bounded. Since $\partial f(P^l)$ is compact, there exists a constant $a_1$ such that 
	\begin{equation}
	\setlength{\abovedisplayskip}{4pt}	
	\setlength{\belowdisplayskip}{4pt}
	\left\|P^{l} + \omega -g(P^{l}) -\mu\right\|< a_1
	\end{equation}
	Define following function
	\begin{equation}
	\tilde V_d(d) =\frac{1}{2} \left\|d\right\|^2
	\end{equation}
	The time derivative of $ \tilde V_d(d) $ along the closed-loop system is 
	\begin{align}
	\dot{\tilde V}_d &=d^{\rm T}(-{{d}} + P^{l} + \omega -g(P^{l}) -\mu)\nonumber\\
	&=-\left\|d\right\|^2+d^{\rm T}( P^{l} + \omega -g(P^{l}) -\mu)\nonumber\\
	&\le -\left\|d\right\|^2+a_1\left\|d\right\|
	=-2\tilde V_d+a_1\sqrt{ 2\tilde V_d}
	\end{align}
	Thus, $ \tilde V_d(d(t)), t\ge 0 $ is bounded, so is $ d(t), t\ge0 $. Similarly, we can also have that $ \varphi ^+(t), t\ge 0 $ and $ \varphi ^-(t), t\ge 0 $ are bounded.
	
	2) By the invariance principle in \cite[Theorem 2]{cortes2008discontinuous}, we know that the trajectory $x(t)$  converges to the largest weakly invariant subset $W^*$ contained in $W:=\{\ x\ |\ \dot V(x)=0\ \}$, i.e., once a trajectory enters this subset, it will never departure from it.
	
	From $ \tilde \omega^{\rm T}D\tilde  \omega =0$, we know $ \tilde  \omega =0 $, i.e., $ \omega(t) = \omega^* $.
	Note that $ (P^l-P^{l*})^{\rm T}\left(g(P^l) - g(P^{l*})\right)>0 $ if $ x\neq x^* $ due to the strict convexity of $ f(P^l) $. Thus, we have $ P^l(t)=P^{l*} $ in the set $W^*$. Moreover, from \eqref{dotV2.50}, we have $ \tilde P^l(t)=CBC^{\rm T}\tilde\phi(t) $, or $ 0= \dot{\tilde P}^l(t)=CBC^{\rm T}\dot{\tilde\phi}(t) $, which implies that $ \dot{\tilde\phi}(t)=0 $. Then, we have $ \dot\mu_j(t)=0 $ from \eqref{Controller2}. Similarly, $ \dot d_j(t)=0 $ by \eqref{Controller1}. Up to now, we know $x(t)$ is constant except $\varphi^-(t), \varphi^+(t)$ for $t\rightarrow\infty$.
	
	Moreover, the equality in \eqref{Property1_projection} holds only when $ \mathcal{P}_{\Omega}(x)=\mathcal{P}_{\Omega}(y) $ or $x=\mathcal{P}_{\Omega}(x)$ and $y=\mathcal{P}_{\Omega}(y) $. Thus, for $ \eta^+(t),\ \eta^-(t), t\rightarrow \infty $, there are four combinations: 
	\begin{enumerate}
		\item  $ \eta^+(t) =\eta^{+*} $ and $ \eta^-(t) =\eta^{-*} $; 
		\item  $ \eta^+(t) =\eta^{+*} $ and $ {\varphi}^-(t)= \mathcal{P}_{\mathbb{R}^{m}_+}\left({\varphi}^{-}(t)\right)=\eta^{-}(t)$, $ {\varphi}^{-*}= \mathcal{P}_{\mathbb{R}^{m}_+}\left({\varphi}^{-*}\right)=\eta^{-*}$; 
		\item  $ {\varphi}^+(t)= \mathcal{P}_{\mathbb{R}^{m}_+}\left({\varphi}^{+}(t)\right)=\eta^{+}(t), {\varphi}^{+*}= \mathcal{P}_{\mathbb{R}^{m}_+}\left({\varphi}^{+*}\right)=\eta^{+*}$ and $ \eta^- =\eta^{-*} $; 
		\item  $ {\varphi}^+(t)= \mathcal{P}_{\mathbb{R}^{m}_+}\left({\varphi}^{+}(t)\right)=\eta^{+}(t), {\varphi}^{+*}= \mathcal{P}_{\mathbb{R}^{m}_+}\left({\varphi}^{+*}\right)=\eta^{+*}$ and $ {\varphi}^-(t)= \mathcal{P}_{\mathbb{R}^{m}_+}\left({\varphi}^{-}(t)\right)=\eta^{-}(t), {\varphi}^{-*}= \mathcal{P}_{\mathbb{R}^{m}_+}\left({\varphi}^{-*}\right)=\eta^{-*}$.
	\end{enumerate}
	
	Thus, $(x(t), P^{l}(t), \eta^+(t), \eta^-(t)) $  converges to equilibrium of the closed-loop system. 
	
	3) Fix any initial state $x(0)$ and consider the trajectory $(x(t), t\geq 0)$ of the closed-loop system.		 
	As $x(t)$ is bounded, there exists an infinite sequence of time instants ${t_k}$ such that $x(t_k)\to\hat x^*$ as $t_k\to\infty$, 
	for some $\hat x^* \in W^*$. Using this specific equlibrium point $\hat {x}^*$ in the definition of $V$, we have 
	\begin{equation}
	\setlength{\abovedisplayskip}{4pt}	
	\setlength{\belowdisplayskip}{4pt}
	\begin{split}
	V^* =\lim\limits_{t\to \infty} V(x(t)) &= \lim\limits_{t_k\to \infty} V(x(t_k)) \nonumber\\
	&=\lim\limits_{x(t_k) \to \hat x^*} V_2\big( x(t_k)\big) = V_2(\hat x^*) = 0 \nonumber
	\end{split}
	\end{equation}
	Here, the first equality uses the  fact that $V(t)$ is nonincreasing in $t$ while lower-bounded, and therefore must render a limit value $ V^* $; the second equality uses the fact that	$t_k$ is the infinite subsequence of $t$; the third equality uses the fact that $x(t)$ is
	absolutely continuous in $t$; the fourth equality is due to the continuity of $V (x)$, and the last equality holds as $\hat {x}^*$ is an equilibrium point of $V $. 
	
	The quadratic part $V_1$  implies that $(\theta, \omega_g, P^l, \mu, \phi, \eta^-,$  $ \eta^+)\to  (\theta^*, \omega_g^*, P^{l*}, \mu^*, \phi^*, \eta^{-*}, \eta^{+*}) $ as $t\to \infty$. Moreover, from \eqref{equilibrium3}, \eqref{equilibrium6} and \eqref{equilibrium7}, we can get the corresponding $ d^*, \varphi^{-*}, \varphi^{+*} $. This completes the proof. 
\end{proof}

\section{Case studies}
\subsection{Test system}
In this section, the IEEE 68-bus New England/New York interconnection test system \cite{Changhong:Design} is utilized to illustrate the performance of the proposed controller. 
The diagram of the 68-bus system is given in Fig.\ref{fig_system}. 
We run the simulation on Matlab using the Power System Toolbox \cite{PST}. Although the linear model is used in the analysis, the simulation model is much more detailed and realistic. The generator includes a two-axis subtransient reactance model, IEEE type DC1 exciter model, and a classical power system stabilizer model. AC (nonlinear) power flows are utilized, including non-zero line resistances. The upper bound of $ P_j^l $ is the load demand value at each bus. Detailed simulation model including parameter values can be found in the data files of the toolbox. 
\begin{figure}[t]
	\centering
	\setlength{\abovecaptionskip}{0pt}
	\includegraphics[width=0.45\textwidth]{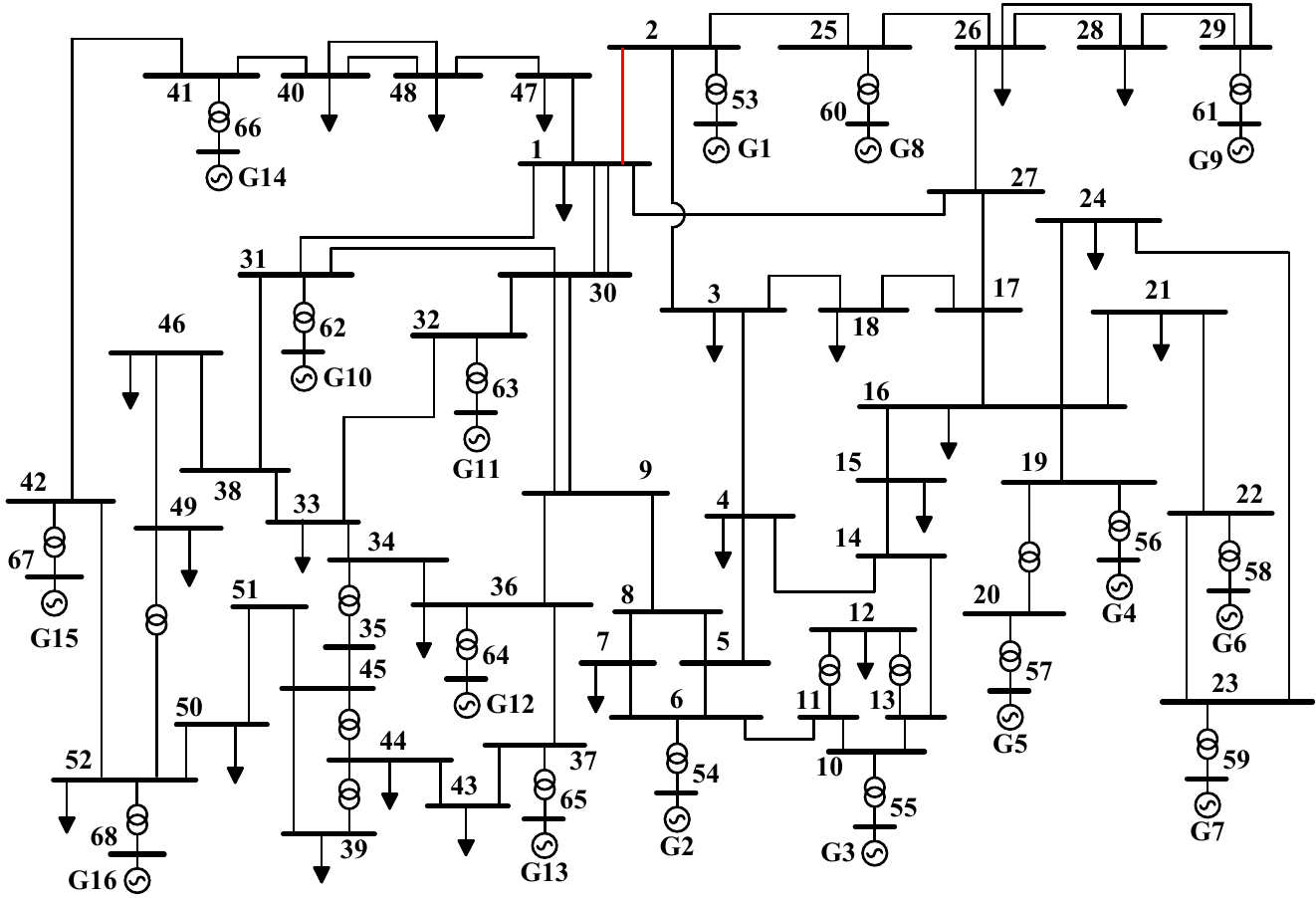}
	\caption{IEEE 68-bus system}
	\label{fig_system}
\end{figure}
The objective function of each controllable load is 
\begin{equation}
\setlength{\abovedisplayskip}{4pt}	
\setlength{\belowdisplayskip}{4pt}
\begin{aligned}
f(P_j^l)=\left\{ \begin{array}{l}
\left(P_j^l\right)^2-0.02,\quad\ \ P_j^l\le -0.2\\
\frac{1}{2}\left(P_j^l\right)^2,\qquad\quad\ \  -0.2<P_j^l\le  0.2\\
\left(P_j^l\right)^2-0.02,\quad\ \ 0.2<P_j^l 
\end{array} \right.
\end{aligned}
\end{equation}
It can be verified that $ f(P_j^l) $ is continuous, strictly convex and nonsmooth. 

\subsection{Simulation results}
We consider the following scenario: at $t=1$s, there is a step change of $(3.5, 3.5, 3.5, 3.5, 3.5, 7)$p.u. load demand at buses 4, 8, 20, 37, 42, and 52 respectively. Neither the original load demand nor its change is known. The load estimate method in Remark \ref{load_estimate} is utilized. 

At first, we do not set limits to the tie-line power. In this subsection, we analyze the dynamic performance of the closed-loop system under the proposed controller OLC. In addition, automatic generation control (AGC) is tested in the same scenario as a benchmark. The setting of AGC is the same as that in \cite{zhao2018distributed}. The frequency dynamics under OLC and AGC are given in Fig.\ref{Frequency}. It is shown that both AGC and OLC can recover the frequency to the nominal value. They also have similar frequency nadir. Compared with AGC, the frequency under OLC has faster convergence speed. 

\begin{figure}[t]
	\centering
	\setlength{\abovecaptionskip}{0pt}
	\includegraphics[width=0.32\textwidth]{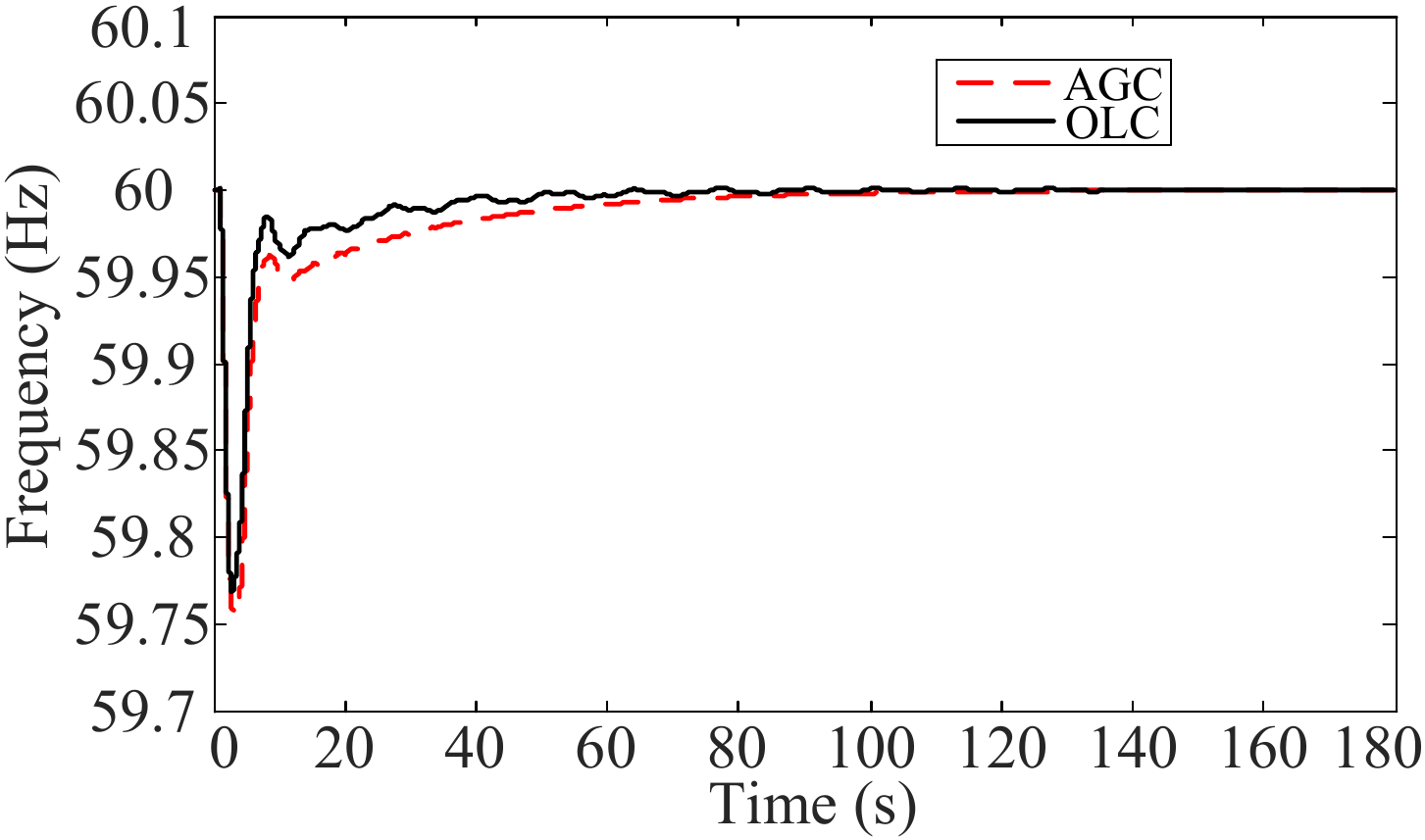}
	\caption{Frequency dynamics under AGC and OLC}
	\label{Frequency}
\end{figure}

The dynamics of $\mu$ and controllable load are illustrated in Fig.\ref{mu_load}. If tie-line power is not considered, then $\eta^{-*}=\eta^{+*}=0$. From \eqref{equilibrium5}, we know that $ CBC^{\rm T}\mu^*=0 $, i.e., $ \mu $ of each bus will converge to the same value. This is shown in the left of Fig.\ref{mu_load}, where $\mu$ of buses $ 1\sim5 $ is given. The right part of Fig.\ref{mu_load} illustrates the dynamics of controllable load at these buses. They also converge to the same value except bus 2 and bus 5 as the objective function is identical. As there is no controllable load on buses 2 and 5, their value is always zero. This validates the correctness of the proposed method.
\begin{figure}[t]
	\centering
	\setlength{\abovecaptionskip}{0pt}
	\includegraphics[width=0.46\textwidth]{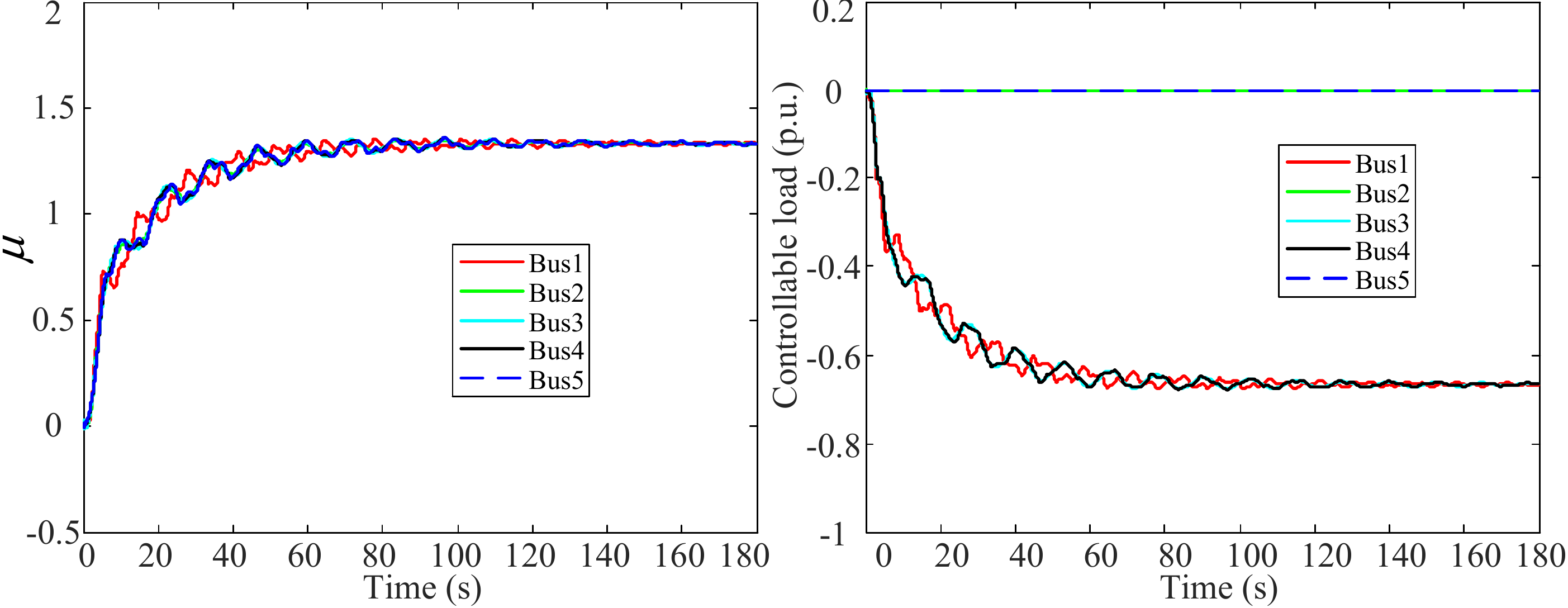}
	\caption{Dynamics of $\mu$ and controllable load}
	\label{mu_load}
\end{figure}

We further take the tie-line limits into account. In this scenario, the active power limit of line $ (1, 2) $ is set to be $0$. Then, its tie-line power dynamics with and without limit is given in Fig.\ref{tie_line}. The active power is $ -2.1 $p.u. if there is no limit. It decreases to zero if the limit is considered, which verifies the effectiveness of the line power control. 

\begin{figure}[t]
	\centering
	\setlength{\abovecaptionskip}{0pt}
	\includegraphics[width=0.32\textwidth]{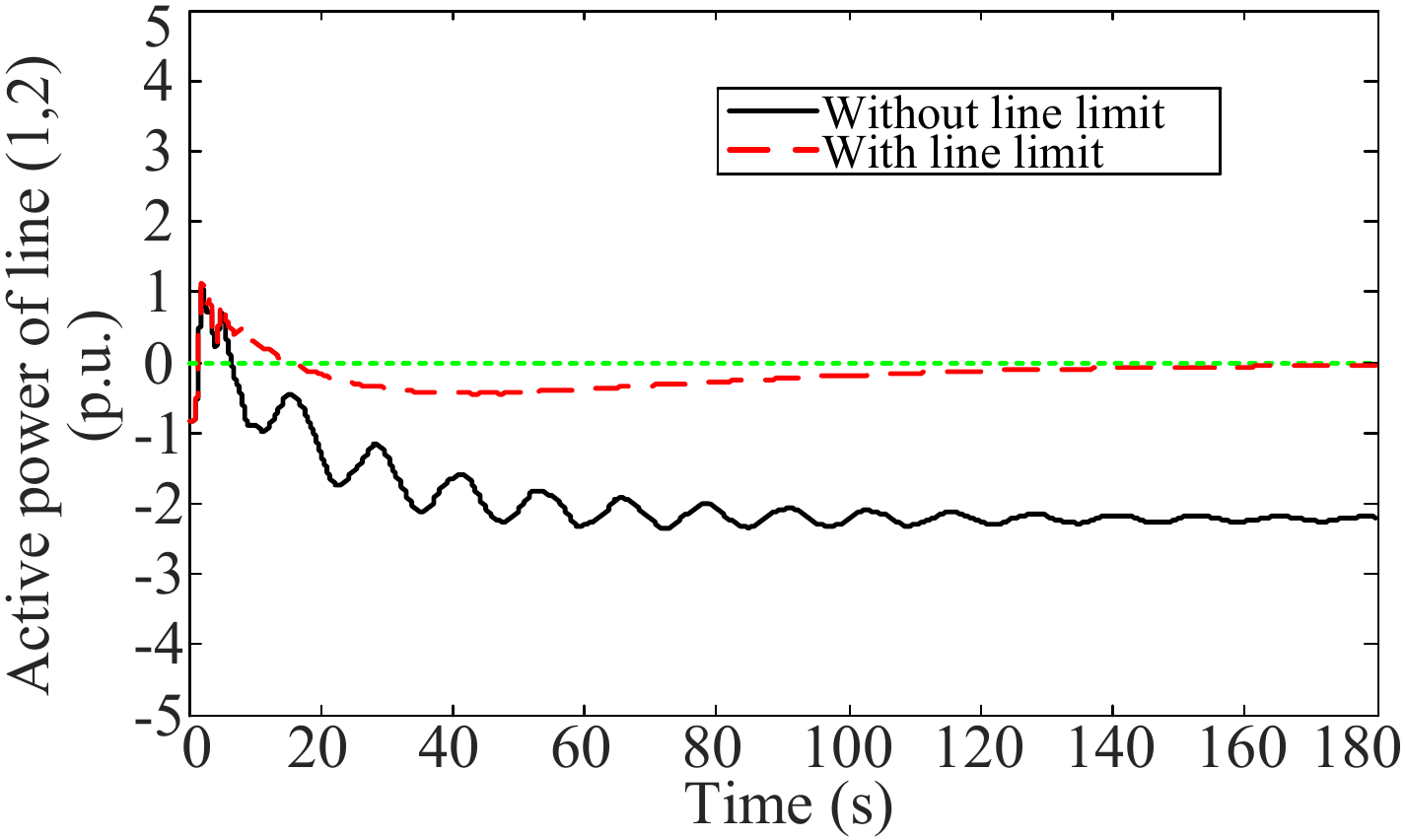}
	\caption{Active power dynamics of line $1$ and controllable load}
	\label{tie_line}
\end{figure}

\begin{figure}[t]
	\centering
	\setlength{\abovecaptionskip}{0pt}
	\includegraphics[width=0.46\textwidth]{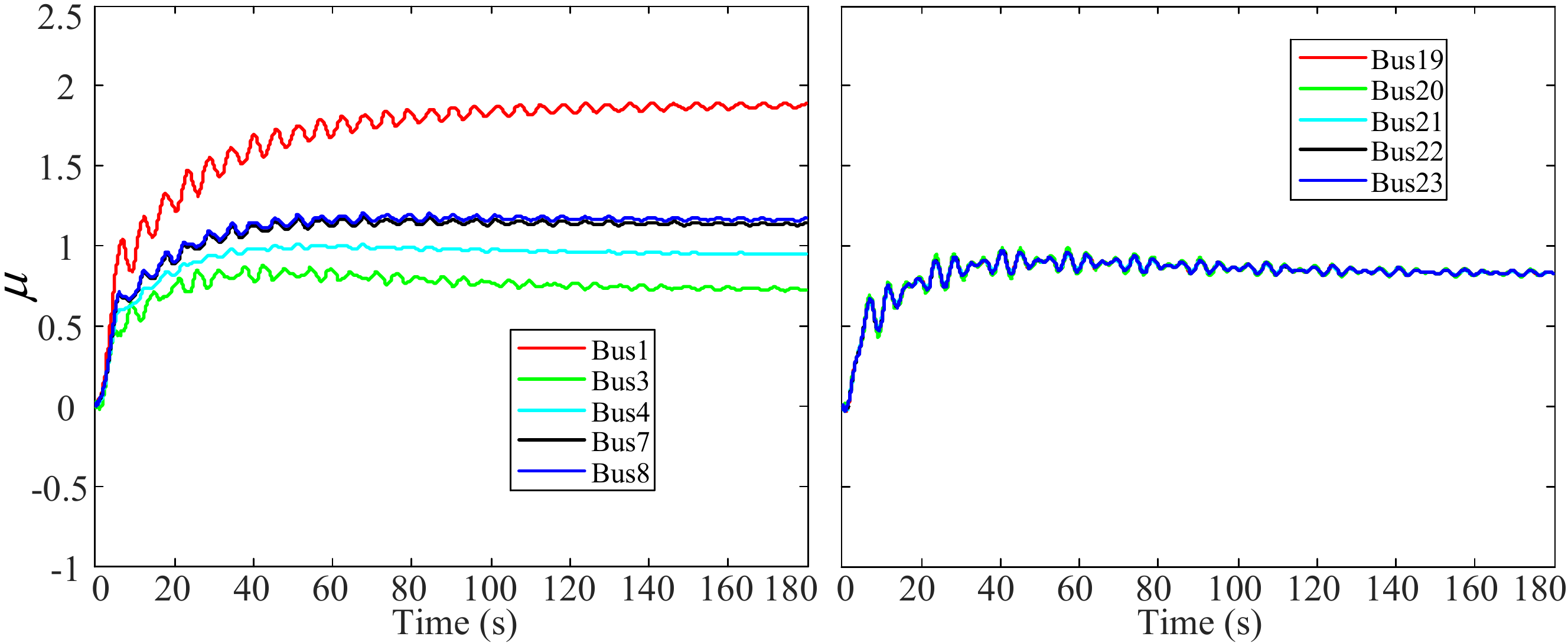}
	\caption{Dynamics of $\mu$ when line power congestion exists}
	\label{mu_congestion}
\end{figure}
If some tie-line limit is reached, the $ \mu $ of buses near from this line will diverge. This is illustrated in the left part of Fig.\ref{mu_congestion}, where $\mu_1,\mu_3,\mu_4$ are all different. We also find that the $ \mu $ of buses far from the line still converges to the same value, which is given in the right part of Fig.\ref{mu_congestion}. It is shown that $ \mu_{19}\sim\mu_{23} $ converge to $0.83$.

\section{Conclusion}
In this paper, we investigate the distributed load frequency control in power systems when the regulation cost is nonsmooth. In our formulation, both capacity limits of controllable load and line flow are considered. A distributed controller is designed, where the Clark generalized gradient is utilized to address the nonsmoothness of the objective function. In addition, we prove the optimality of the equilibrium of the closed-loop system as well as its asymptotic stability. Moreover, it is also proved that the convergence is to a specific point. Finally, numerical experiments on the IEEE 68-bus system show that the frequency is recovered to the nominal value. Compared with conventional AGC, it has faster convergence speed.


\bibliography{mybib}

\end{document}